%% file: elsaesser_new.tex
\documentclass[proceedings]{stacs}
\stacsheading{2009}{373--384}{Freiburg}
\firstpageno{373}

\usepackage{floatflt}
\usepackage{epsfig,amssymb,amsfonts,amsmath,amsthm}
\usepackage{amssymb,amsmath,amsfonts}
\usepackage{shadethm}
\usepackage{multirow}
\usepackage{xspace}
\usepackage{pstricks}
\usepackage[a4paper]{hyperref}

\usepackage{setspace}                  
\usepackage[latin1]{inputenc}

\usepackage{algorithmicx,algorithmicx,algpascal,algpseudocode}
\usepackage{epsfig,amssymb,amsfonts,amsmath,amsthm}
\usepackage{amssymb,amsmath,amsfonts}
\usepackage{multirow}
\usepackage{xspace}
\usepackage{tabularx}
\usepackage{pstricks}

\usepackage{graphicx}
\usepackage{setspace}                  
\usepackage[latin1]{inputenc}

\newcommand{\N}{\mathbb{N}}

\newcommand{\dist}{\operatorname{dist}}
\renewcommand{\deg}{\operatorname{deg}}
\newcommand{\diam}{\operatorname{diam}}

\newcommand{\parti}[1]{\operatorname{part}(#1)}

\renewcommand{\leq}{\leqslant}
\renewcommand{\geq}{\geqslant}

\newcommand{\subs}[1]{\noindent{\bf #1}}

\newcommand{\Oh}{\mathcal{O}}

\newcommand{\MIX}[2]{\mathsf{MIX}_{#1}(#2)}

\newcommand{\Hit}[2]{\mathsf{H}(#1,#2)}
\newcommand{\Com}[2]{\mathsf{C}(#1,#2)}
\newcommand{\Res}[2]{\mathsf{R}(#1,#2)}
\newcommand{\COV}{\mathsf{COV}}

\newcommand{\BLA}{\mathsf{BLA}}

\newcommand{\UFPP}{\mathsf{UFPP}}
\newcommand{\DFPP}{\mathsf{DFPP}}

\newcommand{\pushpar}[2]{\mathsf{RBA}_{#1}(#2)}

\newcommand{\pushseqmod}[2]{\mathsf{\overline{SEQ}}_{#1}(#2)}

\newcommand{\pushparalg}{\mathsf{RBA}}

\newcommand{\pushseqmodalg}{\mathsf{\overline{SEQ}}}

\newcommand{\Pro}[1]{\mathbf{Pr} \left[\,#1\,\right]}

\newcommand{\Ex}[1]{\mathbf{E} \left[\,#1\,\right]}

\renewcommand{\diam}{\operatorname{diam}}

\newcommand{\ratio}{\mathcal{R}}
\renewcommand{\epsilon}{\varepsilon}

\newcommand{\ie}{i.\,e.\xspace}

\theoremstyle{plain}\newtheorem{claim}[thm]{Claim}

\renewcommand{\epsilon}{\varepsilon}
\date{}

\begin{document}

\title[Cover Time and Broadcast Time]{Cover Time and Broadcast Time}

\author[lab1]{R. Els\"asser}{Robert Els\"asser}
\address[lab1]{Institute for Computer Science, University of Paderborn,
33102 Paderborn, Germany}	
\email{elsa@upb.de}  

\author[lab2]{T. Sauerwald}{Thomas Sauerwald}
\address[lab2]{International Computer Science Institute, 1947 Center
Street, Berkeley, CA 94704, U.S.}  
\email{sauerwal@icsi.berkeley.edu}  

\thanks{This work has been partially
 supported by the \textsf{IST} Program of the European Union under contract number 15964 (\textsf{AEOLUS}),
 by the German Science Foundation (\textsf{DFG}) Research Training Group GK-693 of the Paderborn Institute for Scientific Computation (\textsf{PaSCo})
 and by the German Academic Exchange Service (\textsf{DAAD})} 


\begin{abstract}

We introduce a new technique for bounding the cover time of random walks by relating it to the runtime of randomized broadcast. In particular, we strongly confirm for dense graphs the intuition of Chandra et al.~\cite{CRRST97} that ``the cover time of the graph is an appropriate metric for the performance of certain kinds of randomized broadcast algorithms''. In more detail, our results are as follows:

\begin{itemize}
\item For any graph $G=(V,E)$ of size $n$ and minimum degree $\delta$, we have $\mathcal{R}(G)= \Oh(\frac{|E|}{\delta} \cdot \log n)$, where $\mathcal{R}(G)$ denotes the quotient of the cover time and broadcast time. This bound is tight for binary trees and tight up to logarithmic factors for many graphs including hypercubes, expanders and lollipop graphs. 

\item For any $\delta$-regular (or almost $\delta$-regular) graph $G$ it holds that $\mathcal{R}(G) = \Omega( \frac{\delta^2}{n} \cdot \frac{1}{\log n})$. Together with our upper bound on $\mathcal{R}(G)$, this lower bound strongly confirms the intuition of Chandra et al.~for graphs with minimum degree $\Theta(n)$, since then the cover time equals the broadcast time multiplied by $n$ (neglecting logarithmic factors).

\item Conversely, for any $\delta$ we construct almost $\delta$-regular graphs that satisfy $\mathcal{R}(G) = \Oh( \max \{ \sqrt{n},\delta \} \cdot \log^2 n)$. Since any regular expander  satisfies $\mathcal{R}(G) = \Theta(n)$, the strong relationship given above does not hold if $\delta$ is polynomially smaller than $n$.
\end{itemize}

Our bounds also demonstrate that the relationship between cover time and broadcast time is much stronger than the known relationships between any of them and the mixing time (or the closely related spectral gap).

\end{abstract}

\maketitle

\section{Introduction}


\textbf{Motivation.} A \emph{random walk} on a graph is the following process. Starting from a specified vertex, the walk proceeds at each step from its current position to an adjacent vertex chosen uniformly at random. The study of random walks has numerous applications in the design and analysis of algorithms (cf.~\cite{Lovasz93random} for a survey). Two of the most important parameters of random walks are its \emph{mixing time} which is the time until the walk becomes close to the stationary distribution, and its \emph{cover time} which is the expected time required for the random walk to visit all vertices.

Famous combinatorial problems solved by rapidly mixing random walks are, e.g., approximating the permanent and approximating the volume of convex bodies~(cf.~\cite{Lovasz93random} for more details). 
The cover time comes naturally into play when the task is to explore a network, or to estimate the stationary distribution of a graph~\cite{WZ96}. Moreover, the cover time is intimately related to combinatorial and algebraic properties such as the conductance and the spectral gap of the underlying graph \cite{BK89} and thus, bounding the cover time may also lead to interesting combinatorial results.

In this paper, we are particularly interested in the relationship between the cover time of random walks and the runtime of \emph{randomized broadcast}~\cite{FPRU90}. Broadcasting in large networks has various fields of application in distributed computing such as the maintenance of replicated databases or the spreading of information in networks \cite{FPRU90,KSSV00}. Furthermore it is closely related to certain mathematical models of epidemic diseases where infections are spread to some neighbours chosen uniformly at random with some probability. However, in most papers, spreaders are only active in a given time frame, and the question of interest is, whether on certain networks an epidemic outbreak occurs~\cite{KK27,Ne02}. Several threshold theorems involving the basic reproduction number, contact number, and the replacement number have been stated (see~\cite{He00} for a collection of results).

Here, we consider the so-called \emph{randomized broadcast algorithm}~\cite{FPRU90} (also known as \emph{push algorithm}): at the beginning, a vertex $s$ in a graph $G$ knows of some rumor which has to be disseminated to all other vertices. Then, at each time-step every vertex that knows of the rumor chooses one of its neighbors uniformly at random and informs it of the rumor. The advantage of \emph{randomized} broadcast is in its inherent robustness against several kinds of failures (e.g., \cite{FPRU90}) and dynamical changes compared to deterministic schemes that either need substantially more time or can tolerate only a relatively small number of faults \cite{KSSV00}.

\textbf{Related Work.} There is a vast body of literature devoted to the cover time of random walks and we can only point to some results directly related to this paper. Aleliunas et al.~\cite{AKLLR79} initiated the study of the cover time. Amongst other results, they proved that the cover time of any graph $G=(V,E)$ with $n$ vertices is at most $\Oh(n \cdot |E|)$. To obtain this result they proved that the cover time is bounded by the weight of a spanning tree whose edges are weighted according to the commute times between the corresponding vertices. This approach was later refined by Feige~\cite{Fe97} to obtain an upper bound of less than $2 n^2$ for regular graphs. While the spanning tree technique is particularly useful for graphs that have a high cover time~\cite{Fe97}, it vastly overestimates the cover time of e.g., complete graphs.

The seminal work of Chandra et al.~\cite{CRRST97} established a close connection between the \emph{electrical resistance} of a graph and its cover time. This correspondence allows the application of elegant methods from electrical network theory, e.g., the use of short-cut-principles or certain flow-based arguments. Nevertheless, for the computation of the resistance of a given graph other graph-theoretical parameters are often required, e.g., vertex-expansion, number of vertex-disjoint paths or the number of vertices within a certain distance~\cite{CRRST97}.

A wide range of techniques to upper bound the cover time is based on the mixing time of a random walk or the closely related spectral gap. The technique of reducing the cover time to the coupon collector's problem on graphs with low mixing time traces back to Aldous~\cite{Al83} who derived tight bounds on the cover time of certain Cayley graphs. Later, Cooper and Frieze extended this technique to bound the cover time of several classes of random graphs, e.g., \cite{CF05}. The basic idea of this method is that after each mixing time steps, the random walk visits an (almost) randomly chosen vertex. The crux is to deal with the dependencies among the intermediate vertices. Hence, in addition to an upper bound on the mixing time of logarithmic~\cite{CF05} or at least sub-polynomial order~\cite{Al83}, one has to bound the number of returns to the starting vertex within mixing time steps.

A related result was derived by Broder and Karlin~\cite{BK89} who bounded the cover time in terms of the \emph{spectral gap} $1-\lambda_2$, where $\lambda_2$ is the second largest eigenvalue of the transition matrix of the random walk. 

Winkler and Zuckerman~\cite{WZ96} introduced an interesting parameter called \emph{blanket time} which is closely related to the cover time. Here, one asks for the first time-step at which
the observed distribution of the visited vertices approximates the stationary distribution up to a constant factor. Winkler and Zuckerman conjectured that the blanket time is asymptotically the same as the cover time. In~\cite{KKLV00} Kahn et al.~showed that the blanket time is upper bounded by the cover time multiplied by $\Oh((\ln \ln n)^2)$ for any graph.

Most papers dealing with randomized broadcast analyze the runtime on different graph classes. Pittel~\cite{Pi87} proved that the runtime on complete graphs is $\log_2 n + \ln n \pm \Oh(1)$. Feige et al.~\cite{FPRU90}~derived several upper bounds, in particular a bound of $\Oh(\log n)$ for hypercubes and random graphs. We extended the bound of $\Oh(\log n)$ to a certain class of Cayley graphs in~\cite{ES07}. Additionally, we proved that the broadcast time is upper bounded by the sum of the mixing time and an additional logarithmic factor~\cite{S07} (a similar result for a related broadcast algorithm was derived by Boyd et al.~\cite{BGPS06}). However, the mixing time cannot be used for an appropriate lower bound on the broadcast time, as it may overestimate the broadcast time up to a factor of $n$ on certain graphs (cf. Section~\ref{subs:discussion}).

\begin{floatingfigure}[r]{6.8cm} \center{ \hspace{-1.5cm}
\input{figure.pstex_t} 
\caption{All bounds on $\mathcal{R}(G)$ at a glance. $\mathcal{G}_{\delta}$ denotes the class of graphs with $\Delta=\Oh(\delta)$. The blue and red polygons indicate the gap between our lower and upper bounds on $\max_{G \in \mathcal{G}_{\delta}} \mathcal{R}(G)$ and $\min_{G \in \mathcal{G}_{\delta}} \mathcal{R}(G)$, resp.}\label{fig:regular}}
\end{floatingfigure}
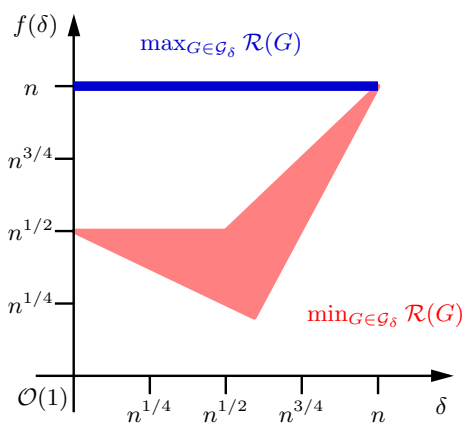
\textbf{Our Results.} We present the first formal results relating the cover time to the broadcast time. In most of them, we will assume that the broadcast and the random walk both start from its respective worst-case initial vertex.
Note that at a first look these processes seem not to be too closely related, since randomized broadcast is a \emph{parallel} process where propagation occurs at \emph{every} informed vertex simultaneously, while a random walk moves "only" from \emph{one} vertex to another \cite{FPRU90}. Nevertheless, Chandra et al.~\cite{CRRST97} mentioned that ``The cover time of the graph is an appropriate metric for the performance of certain kinds of randomized broadcast algorithms''. 
As a consequence of our main results, we obtain a fairly tight characterization of graph classes for which the cover time and broadcast time capture each other. On the positive side, for every graph with minimum degree $\Theta(n)$, the cover time equals the broadcast time multiplied by $n$, up to logarithmic factors (this kind of tightness (up to logarithmic factors) has been frequently considered in the study of random walks, e.g., when studying rapidly mixing Markov chains~\cite{Si92}, or when bounding the cover time~\cite{CRRST97},\cite[Theorem~2.7]{Lovasz93random}.). On the negative side, this strong correspondence does not hold on almost regular graphs, when the degree is substantially smaller than $n$.

In more detail, our results illustrated in Figure~\ref{fig:regular} are as follows. First, we prove that the cover time of any graph with minimum degree $\delta$ is at most $\Oh(\frac{|E|}{\delta} \log n)$ multiplied by the (expected) broadcast time, that is, the quotient $\mathcal{R}(G)$ of the cover time and broadcast time is $\Oh(\frac{|E|}{\delta} \log n)$. This bound is tight up to a constant factor for binary trees and tight up to a logarithmic factor for various graphs including, e.g., expanders, hypercubes and lollipop-graphs. As an application, we use this result to upper bound the cover time of generalized random graphs that are used as a model for real world networks~\cite{CLV03}.

Conversely, we consider the question of lower bounding $\ratio(G)$. By showing that the commute time between two vertices $u,v$ is at least $2 \cdot \dist(u,v)^2$, we obtain that $\ratio(G) = \Omega( \frac{\sqrt{n \log n}}{\Delta})$ for any graph with maximum degree $\Delta$. For constant $\Delta$, this bound is tight for the two-dimensional $\sqrt{n} \times \sqrt{n}$-torus up to logarithmic factors. We move on to improve this bound for denser graphs with $\Delta=\Oh(\delta)$ to $\ratio(G) = \Omega(\frac{\sqrt{n}}{\sqrt{\delta} \log n})$. More importantly, for any graph with $\Delta=\Oh(\delta)$ we establish that $\ratio(G) = \Omega( \frac{\delta^2}{n} \cdot \frac{1}{\log n})$. Together with our upper bound on $\mathcal{R}(G)$, this implies that on any graph with $\delta=\Theta(n)$, cover time and broadcast time (multiplied by $n$) capture each other up to logarithmic factors. 

We complement these positive results by the construction of (almost) $d$-regular graphs for which $\mathcal{R}(G) = \Oh( \max \{ \sqrt{n}, d \} \cdot \log n)$. Since for any $d$-regular expander (graphs for which the spectral gap satisfies $(1-\lambda_2)^{-1}=\Oh(1)$), $\mathcal{R}(G) = \Theta(n)$, the cover time does \emph{not} always capture the performance of randomized broadcast for the class of almost $d$-regular graphs when $d$ is polynomially smaller than $n$.

All of our lower and upper bounds reveal a surprisingly close relationship between the cover time and broadcast time. In particular, upper bounding the cover time in terms of the broadcast time turns out to be as good as (and in some cases much better than) bounding it in terms of the spectral-gap (cf.~Section~\ref{subs:discussion}). From another perspective, we derive a lower bound on the broadcast time in terms of the cover time that nicely complements the existing upper bounds on the broadcast time based on the mixing time~\cite{ES07,S07}. A further novel feature of this work is the use of techniques from electrical network theory to bound the broadcast time. We should note that certain difficulties in applying such methods for the study of randomized broadcast have been mentioned by Feige et al.~\cite{FPRU90}.

%


\section{Notations, Definitions and Preliminaries}\label{sec:randomwalk} 

Throughout this paper, let $G=(V,E)$ be an undirected, simple and connected graph of size $n=|V|$. By $\delta$ and $\Delta$ we denote the minimum and maximum degree of $G$, respectively. For some set $X \subseteq V$, $N(X)$ denotes the set of all neighbors of $x \in X$, and $\deg_{X}(u)$ is the number of edges between $u$ 
and the vertices of $X$.

\subs{Random Walk.} A \emph{random walk} \cite{Lovasz93random} on a graph $G$ starts at a specified vertex $s \in V$ and moves in each step to a neighboring vertex chosen uniformly at random. This can be described by a \emph{transition matrix} $\mathbf{P}$, where $p_{ij} = 1 /\deg(i)$ if $\{i,j\} \in E(G)$, and $p_{ij}=0$ otherwise. Then, the random walk is an infinite sequence of vertices $X_0,X_1,\ldots$, where $X_0:=s$ is the starting point of this random walk, and $X_t$ denotes the vertex visited by the random walk at step $t$. Note that $X_t$ is a random variable with a distribution $\mathbf{p}_s(t)$ on $V(G)$. Denoting by $\mathbf{p}_s(0)$ the unit-vector (regarded as column vector) with $1$ at the component corresponding to $s$ and $0$ otherwise, we obtain the iteration $  \mathbf{p}_s(t+1) =  \mathbf{p}_s(t) \cdot \mathbf{P}$ for every step $t \in \mathbb{N}$. It is well-known that on non-bipartite graphs, $\mathbf{p}_s(t)$ converges for $t\rightarrow \infty$ towards the \emph{stationary distribution} vector $\pi$ given by $\pi(v)= \deg(v) / (2|E|)$. For simplicity, we confine ourselves to non-bipartite graphs in the following. This causes no loss of generality as for general graphs (including bipartite ones) convergence can be ensured easily by using the transition matrix $\frac{1}{2} \mathbf{I} + \frac{1}{2} \mathbf{P}$ (with $\mathbf{I}$ being the identity matrix) instead of $\mathbf{P}$. This change of the transition matrix slows down the mixing time (and the cover time) only by some constant factor~\cite{Lovasz93random,Si92}. 

\subs{Mixing Time and Spectral Gap.} The \emph{mixing time} of a random walk on $G$ is $\MIX{\epsilon}{G} := \max_{s \in V} \min\{ t \in \N :~ \|\mathbf{p}_{s}(t) - \pi \|_{1} \leq \epsilon, X_0=s \}. $
Since $G$ is connected and non-bipartite, the eigenvalues of $\mathbf{P}$ satisfy $\lambda_1 = 1 > \lambda_2 \geq \cdots \geq \lambda_n > - 1$. The following result by Sinclair shows that the spectral gap $1-\lambda_{2}$ captures the mixing time up to logarithmic factors.
\begin{theorem}[\cite{Si92}]\label{thm:sinclairbound}
For any graph $G=(V,E)$ and $\epsilon > 0$,
\[
 \Omega \left(\frac{\lambda_{2}}{1-\lambda_{2}} \cdot \log \bigl(\frac{1}{\epsilon}\bigr) \right) =
\MIX{\epsilon}{G} = \Oh \left(\frac{1}{1-\lambda_{2}} \cdot \left(\log n + \log \bigl(\frac{1}{\epsilon}\bigr) \right) \right). \]
\end{theorem}
%
%

\subs{Commute Time, Resistance and Cover Time.} For two vertices $u,v \in V(G)$, the \emph{hitting time} from $u$ to $v$ is defined as $\Hit{u}{v} := \Ex{\min \{t \in \N \backslash \{0\}: X_t = v, X_0 =u \}}$, \ie, the expected number of steps to reach $v$ from $u$. The \emph{commute time} $\Com{u}{v}$ is defined as the expected number of steps to reach $v$ when starting from $u$ and then returning back to $u$, so, $\Com{u}{v} := \Hit{u}{v} + \Hit{v}{u}$. 
Consider now the graph $G$ as an electrical network where each edge represents a unit resistance. Let $u$ and $v$ be two vertices. Assume that one ampere were injected into vertex $u$ and removed from vertex $v$. Then $\Res{u}{v}$ is the voltage difference between $u$ and $v$ (for more details on electrical networks we refer the reader to \cite{CRRST97,Lovasz93random}), and is related to $\Com{u}{v}$ as follows.

\begin{theorem}[\cite{CRRST97}]\label{thm:commuteresistance}
For any pair of vertices $u,v \in V$, $ \Com{u}{v} = 2 |E| \cdot \Res{u}{v}. $
\end{theorem}

We will mainly be concerned with the \emph{cover time}, which is the expected number of steps a random walk takes to visit all vertices of $G$. Denote by $\COV(s)$ this time for a random walk which starts from $s$, and let $\COV(G):= \max_{s \in V} \COV(s)$. The cover time is related to the maximum commute time by means of $ \frac{1}{2} \cdot \max_{u,v \in V} \Com{u}{v} \leq \COV(G) \leq e^{3} \cdot \max_{u,v \in V} \Com{u}{v} \ln n + n$~\cite{CRRST97}. We restate the following bounds by Feige.
\begin{theorem}[\cite{Fe95b,Fe95}]\label{thm:feigelower}\label{thm:feigeupper}
For any graph, $(1-o(1)) \cdot n \ln n \leq \COV(G) \leq (\frac{4}{27}+o(1)) \cdot n^3$.
\end{theorem}
A corresponding result to Theorem~\ref{thm:sinclairbound} for $\COV(G)$ was given by Broder and Karlin. 
\begin{theorem}[\cite{BK89}]\label{thm:coverspectral} 
For any regular graph $G=(V,E)$,
 $
   \COV(G) = \Oh( \frac{1}{1- \lambda_2} \cdot n \log n).
 $
\end{theorem}

\subs{Randomized Broadcast.} We will consider the relationship between the cover time of random walks and the following \emph{randomized broadcast algorithm} $\pushparalg$ (also known as push algorithm). Assume that at time $t = 0$ a vertex $s$ knows of a rumor which has to be spread to all other vertices. Then, at each time-step $t = 1, 2, \ldots$ every vertex that knows of the rumor chooses a neighbor uniformly at random and informs it of the rumor. Let $I_t$ be the set of informed vertices at time $t$, so $I_0=\{s\}$. The runtime of $\pushparalg$ is denoted by $\pushpar{p}{G} := \max_{s \in V} \min \{ t \in \mathbb{N}: \Pro{I_t = V ~|~ I_{0} = \{s\} } \geq 1 - p \}$ for some given $0< p < 1$. The expected runtime is $\Ex{ \pushpar{}{G}} := \max_{s \in V} \{ \Ex{ \min \{ t \in \mathbb{N}: I_t = V,\, I_{0}=\{s\} \} } \}$. By standard arguments, we have $\Ex{ \pushpar{}{G}} = \Oh( \pushpar{n^{-1}}{G}) = \Oh( \Ex{\pushpar{}{G}} \cdot \log n)$. We remark that $\pushpar{}{G}$ is at least $\max \{ \log_2 n, \diam(G) \}$ on any graph $G$, and $\pushpar{n^{-1}}{G}$ may range from $\Theta(\log n)$ (which is the case for many "nice" graphs) to $\Theta(n \log n)$ (which is the case for the star)~\cite{FPRU90}. Sometimes we also use $\pushpar{}{s,v} := \min \{ t \in \mathbb{N}: v \in I(t) ~|~ I(0) = \{s\} \}$ and $\pushpar{p}{s,v} := \min \{ t \in \mathbb{N}: \Pro{v \in I(t) ~|~ I_{0} = \{s\} } \geq 1 - p \}$ for some specified $0< p < 1$. We will frequently make use of following upper bound of Feige et al.~\cite{FPRU90}.
\begin{theorem}[\cite{FPRU90}]\label{thm:degreebound}
For any graph $G=(V,E)$, $\pushpar{n^{-1}}{G}=\Oh( \Delta \cdot (\log n + \diam(G))$.
\end{theorem}
To compare the cover time with the broadcast time, we define $\ratio(G):=\frac{\COV(G)}{\Ex{\pushpar{}{G}}}$.
\section{Upper Bound on $\mathcal{R}(G)$ and Applications}\label{sec:upperbound}


\subsection{Upper Bound on $\mathcal{R}(G)$}

To prove an upper bound on $\mathcal{R}(G)$, we first prove a general inequality between \emph{first-passage-percolation} times and broadcast times and apply then a result of Lyons et al.~\cite{LPP97} relating first-passage-percolation to the cover time.

\begin{definition}[\cite{FP93,LPP97}]
The \emph{undirected first-passage-percolation} $\UFPP$ is defined as follows. All (undirected) edges $e \in E(G)$ are assigned weights $w(e)$ that are independent exponential random variable with parameter $1$. Specify a vertex $s$. Then the \emph{first-passage-percolation time} from $s$ to $v$ is defined by
$
 \UFPP(s,v) := \inf_{ \mathcal{P}=(s,\ldots,v)} \sum_{e \in \mathcal{P}}
 w(e),
$
where the $\inf$ is over all possible paths from $s$ to $v$ in $G$. Note that $\UFPP(s,s)=0$.
\end{definition}

\begin{theorem}\label{thm:pushpercolation}
For any graph $G=(V,E)$ and $s, v \in V$, $
 \Ex{\UFPP(s,v)} \leq \frac{2}{\delta} \cdot \Ex{ \pushpar{}{s,v} }.
$
\end{theorem}
\begin{proof}
In the proof we derive several (in-)equalities between different percolation and broadcast models. First we introduce a directed version of $\UFPP$, denoted by $\DFPP.$ In this model each undirected edge $\{u,u'\} \in E(G)$ is replaced by two directed edges $(u,u')$ and $(u',u),$ and all directed edges $e$ are assigned weights $w(e)$ that are independent exponential random variable with parameter $1$. Denote by $\DFPP(s,v)$ the corresponding first-passage-percolation time of this directed version.
\begin{lemma}\label{lem:directed}
For any graph $G=(V,E)$ and $s, v \in V$, 
$
 \Ex{\UFPP(s,v)} \leq 2 \cdot \Ex{ \DFPP(s,v)}.
$
\end{lemma}
Next consider another broadcast model denoted by $\pushseqmodalg$. At the beginning, a vertex $s$ knows of a rumor which has to be spread to all other vertices. Once a vertex $u$ receives the rumor at time $t \in \mathbb{R}$, it sends the rumor at each time $t+X_{1,u}$, $t+X_{1,u}+X_{2,u}, \ldots$ to a randomly chosen neighbor, where the $X_{i,u}$ with $i \in \mathbb{N}$ are independent exponential variables with parameter $\deg(u)$. Let $\pushseqmod{}{s,u}$ be the first time when $u$ is informed.

\begin{lemma}\label{lem:samedistribution}
For any $s,v \in V$, $\pushseqmod{}{s,v}$ and $\DFPP(s,v)$ have the same distribution.
\end{lemma}

Finally, our aim is to relate $\pushseqmodalg$ and $\pushparalg$.

\begin{observation}\label{obs:minimalpathpar}
In any execution of $\pushparalg$, there is for each $v \in V$ at least one minimal path $\mathcal{P}_{\min}(s,v)=(s=v_0 \stackrel{D_1}{\rightarrow} v_1 \stackrel{D_2}{\rightarrow} \ldots \stackrel{D_{l-1}}{\rightarrow} v_l=v)$, such that for each $i$, $v_{i}$ sends the rumor $v_{i+1}$ at time $\pushpar{}{s,v_i}+D_{i+1}$, and at this time $v_{i+1}$ becomes informed for the first time.
\end{observation}
Using this observation and a coupling argument, we can prove the following lemma.
\begin{lemma}
For any pair of vertices $s,v \in V$ we have
$
 \Ex{ \pushseqmod{}{s,v}} \leq \frac{\Ex{ \pushpar{}{s,v}}}{\delta}.
$
\end{lemma}
\noindent We are now ready to finish the proof of Theorem
\ref{thm:pushpercolation}. For every pair of vertices $s,v \in V,$
\begin{align*}
\Ex{ \UFPP(s,v)} &\leq 2 \cdot \Ex{ \DFPP(s,v)} =
 2 \cdot \Ex{\pushseqmod{}{s,v}} \leq \frac{2}{\delta} \cdot
 \Ex{\pushpar{}{s,v}}. 
\end{align*}
\end{proof}


\begin{theorem}[\cite{LPP97}]\label{thm:peres}
Let $s,v \in V(G)$ with $s \neq v$. Then, $
 \Res{s}{v} \leq \Ex{\UFPP(s,v)}.
$
\end{theorem}

\noindent Combining the two theorems above we arrive at the main result of this section.

\begin{theorem}\label{thm:uppercover}
For any graph $G=(V,E)$ we have for every pair of vertices $s \neq v$,
\[
 \Com{s}{v} \leq 4 \cdot \frac{|E|}{\delta} \cdot \Ex{\pushpar{}{s,v}},\] and hence
$
 \COV(G) = \Oh \left( \frac{|E|}{\delta} \cdot \log n \cdot \Ex{ \pushpar{}{G} }  \right) \mbox{ or equivalently,~~} 
 \mathcal{R}(G) = \Oh \left( \frac{|E|}{\delta} \cdot \log n \right).
$
\end{theorem}

\subsection{Applications}\label{subs:discussion}

We start by giving examples for which the first inequality of Theorem~\ref{thm:uppercover} is asymptotically tight. For paths and cycles with $n$ vertices, it is well-known that $\max_{s,v} \Com{s}{v}=\Theta(n^2)$~(e.g., \cite{Lovasz93random}) and Theorem~\ref{thm:degreebound} gives $\max_{s,v} \Ex{ \pushpar{}{s,v} } \leq \Ex{\pushpar{}{G}}=\Oh(n)$. Similarly, for lollipop graphs (a complete graph with $2n/3$ vertices attached by a path of length $n/3$), $\max_{s,v} \Com{s}{v}=\Theta(n^2)$~(e.g., \cite{Lovasz93random}) and $\Ex{\pushpar{}{G}}=\Oh(n)$, and therefore the first inequality of Theorem~\ref{thm:uppercover} is also asymptotically tight for this highly non-regular graph.

The following overview in Figure~\ref{fig:examples} is based on~\cite[Chapter 5, p.~11]{AF02}, where we have added the corresponding broadcast times. It can be seen in Figure~\ref{fig:examples} that the second inequality of Theorem~\ref{thm:uppercover} is matched by complete $k$-ary trees with  $k=\Oh(1)$. For complete graphs, expanders and hypercubes, the second inequality is tight up to a factor of $\Oh(\log n)$.  
\begin{figure}[h,t]
\center{
\begin{tabular}{|l|l|l|l|}
\hline
Graph  & $\COV(G)$ & $\Ex{\pushpar{}{G}}$ & $(1-\lambda_2)^{-1}$ \\
\hline
path/cycle & $n^2$~\cite{Lovasz93random} & $n$~(Thm.~\ref{thm:degreebound}) & $n^2$~\cite[Ch.~5, p.~11]{AF02} \\
complete $\Oh(1)$-ary tree & $ n \log^2 n$~\cite[Cor.~9]{Zu92} & $\log n$~(Thm.~\ref{thm:degreebound}) & $n$~\cite[Ch.~5, p.~11]{AF02} \\
complete graph & $n \log n$~\cite{Lovasz93random} & $\log n$~\cite{Pi87} & $1$ \\
expander  & $n \log n$~\cite{BK89} & $\log n$~\cite{S07} & $1$ \\
hypercube  & $n \log n$~\cite{Al83} & $\log n$~\cite{FPRU90} & $\log n$~\cite{Lovasz93random} \\
$\sqrt{n} \times \sqrt{n}$-torus & $n \log^2 n$~\cite[Thm.~4]{Zu92} & $\sqrt{n}$~(Thm.~\ref{thm:degreebound}) & $n$~\cite{Lovasz93random} \\
$K_{n/2} \times K_{2}$ & $n \log n$ & $\log n$~\cite{S07} & $n$ \\
lollipop & $n^3$~\cite{Lovasz93random} & $n$ & $n^2$~\cite[Ch.~5, p.~22]{AF02} \\ \hline
\end{tabular}

\caption{Comparison of the asymptotic order of the cover time, broadcast time and spectral gap of various graph classes. Recall that by Theorem~\ref{thm:sinclairbound}, $(1-\lambda_2)^{-1}$ captures the mixing time up to logarithmic factors.}}\label{fig:examples}
\end{figure}

Let us consider the graph $K_{n/2} \times K_{2}$. One can easily verify that $\COV(G)=\Oh(n \log n)$, $\Ex{\pushpar{}{G}}=\Oh(\log n)$, but $(1-\lambda_2)^{-1}=\Omega(n)$ (and consequently $\MIX{e^{-1}}{G}=\Omega(n)$). Comparing these values with the ones of the complete graph, we see that there are graphs with an optimal cover time and optimal broadcast time, but $(1-\lambda_2)^{-1}$ may vary between $\Theta(1)$ and $\Omega(n)$. Hence
the upper bound on the cover time based on the broadcast time can be a polynomial factor smaller than the corresponding bound (Theorem~\ref{thm:coverspectral}) based on the spectral gap $1-\lambda_2$. On the other hand, the following remark shows that by using the broadcast time instead of the spectral gap, we never lose more than a $\log^2 n$ factor:

\begin{remark}\label{rem:worse}
For any regular graph $G$, the second bound of Theorem~\ref{thm:uppercover} implies
\[
  \COV(G) = \Oh\left(\frac{1}{1-\lambda_2} \cdot n \log^3 n\right).
\]
\end{remark}

\noindent In addition, Theorem~\ref{thm:uppercover} implies directly the following well-known bounds.
\begin{enumerate}
 \item Since $\Ex{\pushpar{}{G}}=\Oh(n)$ for regular graphs \cite[Prop.~1]{ES07}, we obtain $\max_{u,v} \Com{u}{v}=\Oh(n^2)$ for regular graphs~\cite[Ch.~6, Cor.~9]{AF02}.
 \item For bounded degree graphs, $\Ex{\pushpar{}{G}}=\Oh(\diam(G))$~(by Theorem~\ref{thm:degreebound}) implies
 $\max_{u,v} \Com{u}{v}=\Oh(n \diam(G))$~\cite[Ch.~6, Cor.~8]{AF02}.
 \item Since $\max_{u,v} \Ex{ \pushpar{}{u,v}}=\Oh(n)$~\cite{FPRU90}, we obtain $\max_{u,v} \Com{u}{v}=\Oh(n^3)$~\cite[Ch.~6, Thm.~1]{AF02}.
\end{enumerate}

Finally, we give an application of Theorem~\ref{thm:uppercover} to certain power law random graphs (such networks are used to model real world networks~\cite{CLV03}).

\begin{definition}
Given an $n$-dimensional vector $\mathbf{d}=(d_1,d_2,\ldots,d_n)$, the \emph{generalized random graph} $G(\mathbf{d})$ is constructed as follows. Each edge $\{i,j\}, 1 \leq i,j \leq n$ exists with prob.~$\frac{d_i \cdot d_j}{\sum_{k=1}^n d_k}$, independently of all other edges.
\end{definition}

\begin{theorem}[\cite{E06b}] \label{thm:robert}
Let $\mathbf{d}$ be a vector such that for all $i$, $d_i > \log^{c} n$, where $c > 2$ is some constant, and the number of vertices with expected degree $d$ is proportional to $(d - \log^c n)^{-1}$. Then, $G(\mathbf{d})$ satisfies $\pushpar{n^{-1}}{G(\mathbf{d})}=\Oh(\log n)$ with probability $1-o(1)$.
\end{theorem}
Since the number of edges satisfies $|E(G(\mathbf{d}))|=\Oh(n \log^c n)$ with probability $1-o(1)$~\cite{CLV03}, we obtain by combining the theorem above with Theorem~\ref{thm:uppercover}:
\begin{corollary}
For $G(\mathbf{d})$ as in Theorem~\ref{thm:robert} we have $\COV(G(\mathbf{d}))=\Oh(n \log^2 n)$ with probability $1-o(1)$.
\end{corollary}

\section{Lower Bounds on $\mathcal{R}(G)$}\label{sec:lowerbound}
 

\subsection{Sparse Graphs}\label{subs:sparse}

\begin{definition}
Given a graph $G=(V,E)$, a set $\Pi \subseteq E(G)$ is called a \emph{cutset separating $u \in V$ from $v \in V$} if every path from $u$ to $v$ includes an edge of $\Pi$.
\end{definition}
\begin{proposition}[{\cite[p.~59]{LPW06},\cite{N59}}]\label{prop:nash}
For $\{\Pi_i\}_{i=1}^{k}, k \in \mathbb{N}$, being disjoint cutsets separating $u$ from $v$, $\Res{u}{v} \geq \sum_{i=1}^{k} |\Pi_i|^{-1}.$
\end{proposition}
 Zuckerman~\cite{Zu92} proved that for any two vertices $u,v$ on a tree, $\Hit{u}{v} \geq \dist(u,v)^2$. Using Proposition~\ref{prop:nash}, we obtain the following generalization (a similar, but less tight bound follows from a result of \cite{Ca85}).
\begin{corollary}\label{cor:diamsquared}
For any $u,v \in V$ of any graph $G$, $\Com{u}{v} \geq 2 \cdot \dist(u,v)^2. $ On the other hand, there are graphs $G$ and $u,v \in V$ such that $\Hit{u}{v} = \Theta(\dist(u,v)) = o(\dist(u,v)^2)$.
\end{corollary}
We remark that Corollary~\ref{cor:diamsquared} is exact for paths (cf.~\cite{Lovasz93random}). Combining Corollary \ref{cor:diamsquared} with the known bounds from Theorem \ref{thm:feigelower} and Theorem \ref{thm:degreebound} yields:

\begin{proposition}\label{pro:simpleprop}
For any graph $G$ with maximum degree $\delta$, $  \mathcal{R}(G) = \Omega ( \frac{\sqrt{n}}{\Delta} \cdot \sqrt{\log n}).$
\end{proposition}
As demonstrated by the $\sqrt{n} \times \sqrt{n}$-torus where $\pushpar{n^{-1}}{G} = \Theta(\sqrt{n})$ (by Theorem
\ref{thm:degreebound}) and $\COV(G) = \Theta(n \log^2 n)$ \cite{Zu92}, this bound is tight up to a factor of $\log^{3/2} n$ for bounded degree graphs.
The next result improves over Proposition~\ref{pro:simpleprop} for dense graphs.
\begin{theorem}\label{thm:online}
For any graph $G$ with $\Delta=\Oh(\delta)$, $ \mathcal{R}(G) = \Omega ( \frac{\sqrt{n}}{\sqrt{\delta}} \cdot \frac{1}{\log n} ).$
\end{theorem}

\subsection{Dense Graphs}\label{subs:dense}

In this section we present results that are tailored for dense graphs, e.g., graphs with minimum degree $\Theta(n)$. Consider a random walk $X_0=s,X_1,\ldots $ on $G$ starting from
$s$. Denote the number of visits to $u$ until time $t$ as $W_{t}(s,u) := |\{ 0 \leq t' \leq t: ~ X_{t'}=u \}|.$

\begin{definition}[\cite{KKLV00,WZ96}]
Consider a graph $G=(V,E)$ and a random walk starting from $s \in V$. Let
\[ \BLA(s) := \Ex{\min \Bigl\{ t \in \N \, \mid \, ~\forall u \in V:~ \frac{1}{2} \cdot t \pi(u) \leq W_t(s,u) \leq 2 \cdot t \pi(u) \Bigr\}}. \] Then, the \emph{blanket time} of $G$ is defined as $\BLA(G) := \max_{s \in V} \BLA(s).$
\end{definition}


\begin{theorem}[\cite{KKLV00}]\label{thm:blanket}
For any graph $G=(V,E)$, $  \BLA(G) = \Oh(\COV(G) \cdot (\log \log n)^2).
$
\end{theorem}
\noindent We also require the following simple graph-theoretical lemma.
\begin{lemma}\label{lem:twocoverlemma}
For every graph $G$, there is a $2$-cover $X$ of $G$ with $|X| \leq \lceil \frac{n}{\delta} \rceil,$ \ie, there is a set $X \subseteq V$ such that for all $v \in V$ there is an $x \in X$ with $\dist(x,v) \leq 2$.
\end{lemma}
Interestingly, it is known that there are graphs with minimum degree $\frac{n}{2}$ for which every $1$-cover (\ie, dominating set) is of size $\Theta(\log n)$~\cite{AS00}, while the lemma above shows that every such graph has a $2$-cover of constant size. We now prove the main result of Section~\ref{sec:lowerbound}.
\begin{theorem}\label{thm:dense}
For any graph with $\Delta=\Oh(\delta)$, $ \Ex{\pushpar{}{G}} = \Oh ( \frac{1}{\delta} \cdot \BLA(G) + \frac{n^2}{\delta^2} \cdot \log^2 n ). $
\end{theorem}
The following corollary follows immediately from Theorem \ref{thm:dense} and Theorem \ref{thm:blanket}.
\begin{corollary}\label{cor:dense}
For any graph $G=(V,E)$ with $\Delta=\Oh(\delta)$ we have $ \mathcal{R}(G) = \Omega ( \frac{\delta^2}{n} \cdot \frac{1}{\log n} ). $
\end{corollary}
Combining Corollary~\ref{cor:dense} with Theorem~\ref{thm:uppercover} for graphs with minimum degree $\Theta(n)$, we see that the cover time equals the broadcast time multiplied by $n$ up to logarithmic factors. It is worth mentioning that for graphs with $\delta \geq \lfloor \frac{n}{2} \rfloor$, Chandra et al.~\cite[Theorem~3.3]{CRRST97} proved that $\COV(G)=\Theta(n \log n)$. As pointed out by the same authors, $\COV(G)$ may be between $n \log n$ and $\Theta(n^2)$ if $\delta < \lfloor \frac{n}{2} \rfloor$. Now, Corollary~\ref{cor:dense} provides a parameter (the broadcast time) that captures the cover time not only for $\delta \geq \lfloor \frac{n}{2} \rfloor$, but also for $\delta = \Omega(n)$.

\begin{proof}[Proof of Theorem \ref{thm:dense}]
Let us briefly describe the main idea of the proof. We first show that for every vertex $u$ there is a fixed (independent of a concrete execution of $\pushparalg$) set of vertices $Y(u) \subseteq V$ of size at least $\delta/12$ such that $u$ informs each vertex in $Y(u)$ within $\Oh((n/\delta) \cdot \log^2 n)$ steps with high probability. We then establish that if a vertex $u$ informs $v$ in $\Oh((n/\delta) \cdot \log^2 n)$ steps with high probability, then also $v$ informs $u$ in $\Oh((n/\delta) \cdot \log^2 n)$ steps with high probability. Using this fact and Lemma \ref{lem:twocoverlemma} we find that there is a partitioning of $V$ into a constant number of partitions with the following property: once a vertex in such a partition becomes informed, the whole partition becomes informed within $\Oh((n/\delta) \cdot \log^2 n)$ steps. Finally, we use a coupling between the random walk and the broadcast algorithm to show that if the random walk covers the whole graph quickly, then the rumor will also be quickly propagated from one partition to the other partitions. The formal proof follows.
\begin{lemma}\label{lem:localexpansion}
For each $u \in V$ there is a set $Y(u) \subseteq V$ (independent of the execution of $\pushparalg$) of size at least $\delta/12$ such that for every $v \in Y(u)$, $\pushpar{n^{-4}}{u,v} \leq 16 C_1 \frac{n}{\delta} \log^2 n ,$ where $C_1 > 0$ is some constant.
\end{lemma}

\begin{lemma}\label{lem:uvvusymmetry}
For any two vertices $u,v$ in a graph $G$ with $\Delta=\Oh(\delta),$ $ \pushpar{n^{-4}}{v,u} \leq C_2 \cdot (\pushpar{n^{-4}}{u,v} + \log n),$ where $C_2 > 0$ is some constant.
\end{lemma}


Consider the undirected auxiliary graph $\widehat{G}=(\widehat{V},\widehat{E})$ defined as follows: $\widehat{V}:=V$ and $\{u,v\} \in \widehat{E}$ iff \[ \max \{ \pushpar{n^{-4}}{u,v}, \pushpar{n^{-4}}{v,u} \} \leq C_2 \cdot \left(16 C_1  \frac{n}{\delta}  \log^2 n + \log n \right). \] By the two lemmas above, $\delta(\widehat{G}) \geq \delta/12$. Hence Lemma \ref{lem:twocoverlemma} implies the existence of a $2$-cover $\{u_1,u_2,\ldots,u_{k}\}$, $k \leq \lceil n/\delta \rceil$, of $\widehat{G}$. Therefore, the sets $U_i := \{ v \in \widehat{V} ~|~\dist_{\widehat{G}}(v,u_i) \leq 2 \}, 1 \leq i \leq k$ form a (possibly non-disjoint) partitioning of $\widehat{V}.$ Take a disjoint partitioning $V_1, V_2, \ldots, V_k$ such that for every $1 \leq i \leq k$, $V_i \subseteq U_{i}$. Consider now the directed graph $G':=(V',E')$ with $V' := \{ V_1,
V_2, \ldots, V_{k}   \} $ and 
\begin{align*}
  E' &:= \left\{ (V_i,V_j) ~|~ \exists u \in V_i,~1 \leq t \leq 4 \cdot \frac{\deg(u)}{2|E|} \cdot \BLA(G): ~ N_{t,u} \in V_j \right\},
  \end{align*}
where $N_{t,u} \in N(u)$ is the vertex to which the random walk moves after the $t$-th visit of $u$.
\begin{claim}
Let $s \in V_{i}.$ With prob. $1/2$, there is a path from $V_i$ to every $V_j$ in $G'$.
\end{claim}

Reconsider now the partitioning $V_1,V_2, \ldots, V_{k},$ $k \leq \lceil n/\delta \rceil,$ of $\widehat{V}=V$. Let $\parti{u}$ be the function which assigns a vertex $u$ the index of its partition. Let $\mathcal{B}$ be the event that $
 \forall u \in V: V_{\parti{u}} \subseteq I_{\pushpar{}{s,u}+\Oh (\frac{n}{\delta} \log^2 n )}
$ holds, \ie, for all $u \in V$ it holds that once $u$ is informed, the partition $V_{\parti{u}}$ becomes completely informed within further $\Oh \left(\frac{n}{\delta} \log^2 n \right)$ steps. Fix some arbitrary vertex $u \in V$ and consider another vertex $w \in V_{\parti{u}}$. By definition of $\widehat{G}$ and Lemma~\ref{lem:uvvusymmetry}, there is path of length at most $4$ from $u$ to $w$ in $\widehat{G}$. Hence once $u$ is informed, $w$ becomes informed within the next $\Oh(\frac{n}{\delta} \log^2 n)$ steps with probability $1-4 n^{-4}$. Applying the union bound over $u \in V$ and $w \in V_{\parti{u}}$, we get $\Pro{ \mathcal{B} } \geq 1 - 4 n^{-2}$.
\begin{claim}
Conditioned on the events $\mathcal{A}$ and $\mathcal{B}$, all vertices of $G$ become informed after
$
 \Oh \left(   \frac{1}{\delta} \cdot \BLA(G) + \frac{n^2}{\delta^2} \cdot \log^2 n \right)
$
steps.
\end{claim}
To finish the proof of the Theorem, we apply the union bound to get $ \Pro{ \mathcal{A} \wedge \mathcal{B} } \geq 1 - \frac{1}{2} - 4 n^{-2}. $ So, with probability larger than $1/3$, all vertices of $G$ become informed after at most $\Oh (\frac{1}{\delta} \cdot \BLA(G) + \frac{n^2}{\delta^2} \cdot \log^2 n )$ steps. Thus for every $k \in \mathbb{N}$, we succeed after $\Oh (k \cdot (\frac{1}{\delta} \cdot \BLA(G) + \frac{n^2}{\delta^2} \cdot \log^2 n) )$ steps with probability $1-(2/3)^k$ and hence the expected broadcast time is $\Oh (\frac{1}{\delta} \cdot \BLA(G) + \frac{n^2}{\delta^2} \cdot \log^2 n )$.
\end{proof}

\subsection{Discussion} \label{subs:construction}

We first complement the lower bounds on $\mathcal{R}(G)$ by some concrete graphs.
By a construction based on Harary graphs~\cite{Har62} and the two-dim. torus we obtain the following.
\begin{theorem}\label{thm:construction}
For any $\sqrt{n} \leq d \leq n-1$, there is a $d$-regular graph $G$ with $\mathcal{R}(G) = \Oh( d \cdot \log n)$. Moreover, for any $1 \leq d \leq \sqrt{n}$ there is a graph with minimum degree $d$ and maximum degree $d+1$ such that $\mathcal{R}(G) = \Oh( \sqrt{n} \cdot \log^2 n)$.
\end{theorem}

While for certain degrees, a small polynomial gap remains between the examples of Theorem~\ref{thm:construction} and the bounds of Theorem~\ref{thm:online} and Theorem~\ref{thm:dense}~(cf.~Figure~\ref{fig:regular}), the quotient between cover time and diameter is minimized up to logarithmic factors by these examples.

\begin{proposition}\label{pro:coverdiam}
For any graph $G$ with $\Delta=\Oh(\delta)$, $ \frac{\COV(G)}{\diam(G)} = \Omega( \max \{ \sqrt{n},\delta \} \cdot \sqrt{\log n}).$
\end{proposition}
So far, in all considered graphs with a (nearly) optimal cover times and high broadcast time, the latter was caused by a large diameter. Therefore, one could try to throw in the lower bounds on $\diam(G)$ and ask the following question: Does $\COV(G)=\Oh(\operatorname{polylog}(n) \cdot \max\{n \log n, \diam(G)^2 \}) \Leftrightarrow \Ex{\pushpar{}{G}} = \Oh(\operatorname{polylog}(n) \cdot \max\{ \diam(G), \log n\})$ hold? The answer is that both directions can be refuted by counter-examples, even for graphs where minimum and maximum degree coincide (up to constant factors).

\section{Conclusion}\label{sec:conclude}
Inspired by the intuition of Chandra et al.~\cite{CRRST97} about the relationship between cover time of random walks and the runtime of randomized broadcast, we devised the first formal results relating both times. 
As our main result in Section~\ref{sec:upperbound}, we proved that the cover time of any graph $G$ is upper bounded by $\Oh(\frac{E}{\delta} \log n)$ times the broadcast time. This result is tight for many graphs (at least up to a factor of $\log n$) and gives an upper bound on the cover time that is at least as good (and in certain cases much tighter than) the previous bound based on the spectral gap~\cite{BK89}. Moreover, this result implies several classic bounds on the cover time and an almost optimal upper bound on the cover time of certain random graphs that are used to model real world networks.
In Section~\ref{sec:lowerbound} we derived lower bounds on the ratio between the cover time and broadcast time. Together with our upper bound of Section~\ref{sec:upperbound}, we established a surprisingly strong correspondence between the cover time and broadcast time on dense graphs. This positive result was complemented by the construction of certain graphs to demonstrate that this strong correspondence cannot be extended to sparser graphs. 
Nevertheless, our lower and upper bounds show that the relationship between cover time and broadcast time is substantially stronger than the relationship between any of these parameters and the mixing time (or the closely related spectral gap). In particular, our findings provide evidence for the following hierarchy for regular graphs:
\[
 \mbox{low mixing time} \Rightarrow \mbox{low broadcast time} \Rightarrow \mbox{low cover
 time},
\]
which extends the following known relations: $ \mbox{low mixing time} \Rightarrow \mbox{low cover time} $~(\cite{Al83,BK89,CF05}) and $\mbox{low mixing time} \Rightarrow \mbox{low broadcast time}$~(\cite{BGPS06,ES07,S07}).

\addcontentsline{toc}{chapter}{\protect\numberline{B}{Bibliography}}
\bibliography{cover,mybib,diss,Propp,bib3,bib2,rainer2}

\bibliographystyle{plain}

\end{document}

%% file: figure.pstex_t
\begin{picture}(0,0)%
\includegraphics{figure.pstex}%
\end{picture}%
\setlength{\unitlength}{3947sp}%
\begingroup\makeatletter\ifx\SetFigFont\undefined%
\gdef\SetFigFont#1#2#3#4#5{%
  \reset@font\fontsize{#1}{#2pt}%
  \fontfamily{#3}\fontseries{#4}\fontshape{#5}%
  \selectfont}%
\fi\endgroup%
\begin{picture}(2996,2621)(-648,-2917)
\put(377,-546){\makebox(0,0)[lb]{\smash{{\SetFigFont{9}{10.8}{\sfdefault}{\mddefault}{\updefault}{\color[rgb]{0,0,.82}$\max_{G \in \mathcal{G}_{\delta}}  \ratio(G) $}%
}}}}
\put(-273,-432){\makebox(0,0)[b]{\smash{{\SetFigFont{9}{10.8}{\familydefault}{\mddefault}{\updefault}{\color[rgb]{0,0,0}$f(\delta)$}%
}}}}
\put(436,-2872){\makebox(0,0)[b]{\smash{{\SetFigFont{9}{10.8}{\sfdefault}{\mddefault}{\updefault}{\color[rgb]{0,0,0}$n^{1/4}$}%
}}}}
\put(908,-2872){\makebox(0,0)[b]{\smash{{\SetFigFont{9}{10.8}{\familydefault}{\mddefault}{\updefault}{\color[rgb]{0,0,0}$n^{1/2}$}%
}}}}
\put(1380,-2872){\makebox(0,0)[b]{\smash{{\SetFigFont{9}{10.8}{\sfdefault}{\mddefault}{\updefault}{\color[rgb]{0,0,0}$n^{3/4}$}%
}}}}
\put(1854,-2872){\makebox(0,0)[b]{\smash{{\SetFigFont{9}{10.8}{\familydefault}{\mddefault}{\updefault}{\color[rgb]{0,0,0}$n$}%
}}}}
\put(-302,-2191){\makebox(0,0)[b]{\smash{{\SetFigFont{9}{10.8}{\sfdefault}{\mddefault}{\updefault}{\color[rgb]{0,0,0}$n^{1/4}$}%
}}}}
\put(-302,-1283){\makebox(0,0)[b]{\smash{{\SetFigFont{9}{10.8}{\sfdefault}{\mddefault}{\updefault}{\color[rgb]{0,0,0}$n^{3/4}$}%
}}}}
\put(-302,-829){\makebox(0,0)[b]{\smash{{\SetFigFont{9}{10.8}{\familydefault}{\mddefault}{\updefault}{\color[rgb]{0,0,0}$n$}%
}}}}
\put(-229,-2766){\makebox(0,0)[b]{\smash{{\SetFigFont{9}{10.8}{\familydefault}{\mddefault}{\updefault}{\color[rgb]{0,0,0}$\Oh(1)$}%
}}}}
\put(-302,-1737){\makebox(0,0)[b]{\smash{{\SetFigFont{9}{10.8}{\familydefault}{\mddefault}{\updefault}{\color[rgb]{0,0,0}$n^{1/2}$}%
}}}}
\put(2248,-2822){\makebox(0,0)[b]{\smash{{\SetFigFont{9}{10.8}{\familydefault}{\mddefault}{\updefault}{\color[rgb]{0,0,0}$\delta$}%
}}}}
\put(1902,-2242){\makebox(0,0)[b]{\smash{{\SetFigFont{9}{10.8}{\familydefault}{\mddefault}{\updefault}{\color[rgb]{1,0,0}$\min_{G \in \mathcal{G}_{\delta}}  \ratio(G) $}%
}}}}
\end{picture}%